\documentclass[10pt,conference]{IEEEtran}

\usepackage[utf8]{inputenc}
\usepackage[T1]{fontenc}
\usepackage{amsmath,amssymb,amsthm,mathtools}
\usepackage{bbm}
\usepackage{hyperref}
\usepackage{graphicx}
\usepackage{microtype}
\usepackage{xcolor}
\usepackage{tikz}
\usetikzlibrary{positioning}

\newtheorem{definition}{Definition}
\newtheorem{theorem}{Theorem}
\newtheorem{proposition}{Proposition}
\newtheorem{lemma}{Lemma}
\newtheorem{remark}{Remark}

\newcommand{\Nat}{\mathbb{N}}
\newcommand{\Real}{\mathbb{R}}

\newcommand{\Tick}{\mathsf{Tick}}
\newcommand{\cfg}{\mathsf{cfg}}
\newcommand{\Lp}{\mathcal{L}}

\newcommand{\Pool}{\mathsf{Pool}}
\newcommand{\State}{\mathsf{State}}
\newcommand{\Next}{\mathsf{Next}}
\newcommand{\invar}{\mathsf{Inv}}
\newcommand{\kbar}{\overline{k}}

\begin{document}

\title{Formal State-Machine Models for Uniswap v3 Concentrated-Liquidity AMMs: Priced Timed Automata, Finite-State Transducers, and Provable Rounding Bounds}

\author{
    \IEEEauthorblockN{Julius Tranquilli\IEEEauthorrefmark{2} and Dr. Naman Gupta\IEEEauthorrefmark{2}}
    \vspace{4pt}
    \IEEEauthorblockA{\IEEEauthorrefmark{2}Black Meridian Research
    \vspace{4pt}
    \\\textcolor{blue}{julius@blackmeridianresearch.com} 
    \\\textcolor{blue}{naman@blackmeridianresearch.com}}
}

\maketitle

\begin{abstract}
Concentrated-liquidity automated market makers (CLAMMs), as exemplified by Uniswap v3, are now a common primitive in decentralized finance frameworks. Their design combines continuous trading on constant-function curves with discrete tick boundaries at which liquidity positions change and rounding effects accumulate. While there is a body of economic and game-theoretic analysis of CLAMMs, there is negligible work that treats Uniswap v3 at the level of formal \emph{state machines} amenable to model checking or theorem proving.

In this paper we propose a formal modeling approach for Uniswap v3-style CLAMMs using (i) networks of priced timed automata (PTA), and (ii) finite-state transducers (FST) over discrete ticks. Positions are treated as stateful objects that transition only when the pool price crosses the ticks that bound their active range. We show how to encode the piecewise constant-product invariant, fee-growth variables, and tick-crossing rules in a PTA suitable for tools such as \textsc{UPPAAL}, and how to derive a tick-level FST abstraction for specification in TLA+.

We define an explicit tick-wise invariant for a discretized, single-tick CLAMM model and prove that it is preserved up to a tight additive rounding bound under fee-free swaps. This provides a formal justification for the ``$\epsilon$-slack'' used in invariance properties and shows how rounding enters as a controlled perturbation. We then instantiate these models in TLA+ and use TLC to exhaustively check the resulting invariants on structurally faithful instances, including a three-tick concentrated-liquidity configuration and a bounded no-rounding-only-arbitrage property in a bidirectional single-tick model. We discuss how these constructions lift to the tick-wise structure of Uniswap v3 via virtual reserves, and how the resulting properties can be phrased as PTA/TLA+ invariants about cross-tick behaviour and rounding safety.
\end{abstract}

\section{Introduction}

Automated market makers (AMMs) are a core building block of decentralized finance, enabling on-chain swaps via constant-function pricing rather than order books. Uniswap v3 introduced \emph{concentrated liquidity}, allowing liquidity providers (LPs) to allocate capital to bounded price intervals rather than over the full price range~\cite{adams2021uniswapv3}. This design increases capital efficiency but complicates both economic analysis and security reasoning: liquidity is now piecewise, tick-driven, and stateful.

Existing work on Uniswap v3 has focused primarily on economic aspects: strategic liquidity provision~\cite{fan2023strategic}, expected fees and impermanent loss~\cite{hashemseresht2022concentrated}, or the exact liquidity mathematics of ticks~\cite{elsts2021liquiditymath}. Simultaneously, there is increasing interest in mechanized reasoning about AMMs, for example the Lean 4 formalization of constant-product AMMs by Pusceddu and Bartoletti~\cite{pusceddu2024leanamm}. However, there is no work that captures Uniswap v3 as an explicit \emph{transition system} suitable for off-the-shelf model checkers such as \textsc{UPPAAL} or TLA+ model checking.

This paper argues that Uniswap v3 and related CLAMMs are viewed as:
\begin{itemize}
  \item a network of \emph{priced timed automata} (PTA)~\cite{bengtsson2004timedautomata} capturing both discrete tick-crossing events and temporal aspects (e.g.\ gas cost budgets, time-weighted liquidity measures), and
  \item a \emph{finite-state transducer} (FST) on discrete ticks and bounded balances, expressing swap and liquidity actions as input symbols and pool/LP outputs as output symbols.
\end{itemize}

Within this framework we sketch a verification agenda that includes:
\begin{enumerate}
  \item \textbf{Invariance:} a pool-specific constant-function quantity remains invariant under swaps and tick-crossings, modulo bounded rounding error.
  \item \textbf{Reachability:} checking whether certain adverse price states (e.g.\ extreme ticks) are reachable under specified adversary and LP strategies.
  \item \textbf{Rounding safety:} establishing the absence of round-trip trading cycles that extract positive value purely from rounding behaviour, in the spirit of economic no-arbitrage conditions.
\end{enumerate}

We construct a core object: a tick-wise invariant $K_i$ and its behaviour under discrete swaps. We define and analyze a simplified, single-tick CLAMM model that captures the arithmetic structure (constant product plus flooring) without the complication of multiple ticks. For this model we prove an explicit bound on the deviation of $XY$ (the usual constant-product invariant) due to rounding, for arbitrary fee-free swaps within the tick.

\paragraph*{Contributions.}
The contributions of this paper are:
\begin{itemize}
  \item A self-contained state-machine model of a Uniswap v3 pool and its positions in terms of PTA and FSTs, parameterized by ticks and bounded balances.
  \item A concrete definition of a tick-wise invariant $K_i$ consistent with Uniswap-style constant-product math, together with a \emph{proved} bound on its change under discretized swaps in a single-tick model. Mathematically this bound is elementary; the contribution lies in turning it into an explicit, parameterised \emph{rounding core} (Definition~\ref{def:rounding-core}, Theorem~\ref{thm:rounding-core-epsilon}) that can be instantiated in mechanised models.
  \item Formal definitions of invariance and safety properties relevant to CLAMM correctness, including a ``$k$-invariance across ticks'' property and a notion of rounding-arbitrage-freeness, together with fully explored TLA+ instances that check these properties under explicit bounds.
  \item A mapping of this model into \textsc{UPPAAL} (for PTA-based model checking) and into TLA+ (for discrete state exploration), with example queries and invariants, showing how the rounding core and its bound thread through the specifications and support exhaustive exploration for bounded parameter regimes.
  \item A discussion of how this approach relates to and complements existing theorem-proving work on AMMs in Lean and general surveys of timed automata in security~\cite{arcile2022timedsecurity}, and how the same templates apply to constant-product AMMs beyond Uniswap v3.
\end{itemize}

\section{Background}

\subsection{Uniswap v3 Concentrated Liquidity}

At a high level, a Uniswap v3 pool manages reserves $(x,y)$ of two ERC-20 tokens $(A,B)$ and supports swaps subject to a constant-function pricing rule, but with liquidity \emph{concentrated} in one or more price intervals~\cite{adams2021uniswapv3,elsts2021liquiditymath}. Price is represented via a discrete \emph{tick} index $i \in \mathbb{Z}$ and a geometric tick-to-price mapping
\[
  P_i = \lambda^i,\quad \lambda > 1,
\]
with $\lambda \approx 1.0001$ in the deployed protocol.

Each LP position is specified by a lower tick $i_{\ell}$, an upper tick $i_u$, and a liquidity amount $L > 0$. The position's liquidity contributes to the pool's \emph{active liquidity} only when the current tick $i_c$ lies within the half-open interval $[i_{\ell}, i_u)$.

Inside a single tick interval $[P_i, P_{i+1})$ with active liquidity $L$, the pool behaves like a standard constant-product AMM in terms of the relationship between price and reserves, but expressed in terms of $\sqrt{P}$ and $L$ rather than explicit reserves~\cite{elsts2021liquiditymath}. When swaps drive the price to the boundary $P_{i+1}$, a \emph{tick crossing} occurs, adjusting active liquidity by the net liquidity of all positions whose lower or upper bounds equal $i+1$.

Economically-focused work has analyzed fees, impermanent loss, and optimal LP strategies under this design~\cite{fan2023strategic,hashemseresht2022concentrated}.

\subsection{Timed and Priced Timed Automata}

Timed automata, introduced by Alur and Dill~\cite{alur1994timed} and later surveyed by Bengtsson and Yi~\cite{bengtsson2004timedautomata}, extend finite-state automata with real-valued clocks that evolve continuously with time. Guards involving clock inequalities control when transitions may fire; invariants on locations bound dwelling time.

Priced timed automata (PTA) further attach costs (or prices) to locations and transitions, enabling model checking of temporal properties that involve accumulated resource usage, such as ``minimal expected cost'' or ``maximal reward''~\cite{david2011priced}. Tools such as \textsc{UPPAAL} and its SMC extension support networks of PTA and queries in a fragment of timed computation tree logic (TCTL)~\cite{david2015uppaalsmc}.

Recent surveys show that timed automata are a natural fit for security protocols and resource-sensitive systems, capturing both functional and temporal/quantitative properties~\cite{arcile2022timedsecurity}. CLAMMs are similarly resource-sensitive: tokens and fees are conserved (up to rounding), while gas consumption and time in-range are important quantities.

\subsection{TLA+ and Model Checking of Blockchain Protocols}

TLA+ is a state-based formal specification language based on actions and temporal logic~\cite{lamport2002tla}. It has been used to model and verify blockchain protocols such as cross-chain swaps and consensus algorithms~\cite{nehai2022crosschain,ouyang2019cbccasper}. In TLA+, systems are described via a state space of variable valuations and a \emph{next-state relation} $\Next$, together with invariants and liveness properties.

A CLAMM like Uniswap v3 fits this style of modeling: pool reserves, active liquidity, fee-growth variables, and positions all contribute to a finite (or finitely bounded) state vector. Swaps, mint-burn operations, and tick crossings are actions.

\section{Uniswap v3 as a State Machine}

We now define a simplified operational model of a Uniswap v3-style CLAMM, separating out the discrete tick dimension and the continuous-time dimension.

\subsection{Discrete Pool State}

We fix:
\begin{itemize}
  \item a finite tick range $\Tick = \{i_{\min}, \dots, i_{\max}\} \subset \mathbb{Z}$,
  \item a tick-to-price mapping $P : \Tick \to \Real_{>0}$, $P_i = \lambda^i$,
  \item a finite index set of positions $\mathcal{I} = \{1,\dots,N\}$.
\end{itemize}

\begin{definition}[Pool Configuration]
A \emph{pool configuration} is a tuple
\[
  \cfg = (i_c, X, Y, L_{\mathrm{act}}, \Lp, F_x, F_y)
\]
where
\begin{itemize}
  \item $i_c \in \Tick$ is the current tick index,
  \item $X,Y \in \Nat$ are token reserves of $A$ and $B$ in minimal units,
  \item $L_{\mathrm{act}} \in \Nat$ is the active liquidity at tick $i_c$,
  \item $\Lp : \mathcal{I} \to \Nat \times \Tick \times \Tick$ maps a position index $k$ to $(L_k, i_{\ell}^k, i_{u}^k)$ such that $i_{\ell}^k < i_{u}^k$ and $L_k \ge 0$,
  \item $F_x, F_y \in \Nat$ are fee-growth variables tracking cumulative fees in token units, as in the Uniswap v3 whitepaper~\cite{adams2021uniswapv3}.
\end{itemize}
\end{definition}

The active liquidity $L_{\mathrm{act}}$ is derived from positions:
\[
  L_{\mathrm{act}} = \sum_{k \in \mathcal{I}: i_{\ell}^k \le i_c < i_{u}^k} L_k.
\]

In the actual protocol, reserves and prices are represented in fixed-point formats with specific Q64.96 encodings~\cite{elsts2021liquiditymath}; we abstract this as integer amounts with a fixed scale factor and a rounding operator $\lfloor \cdot \rceil$ capturing rounding to the nearest representable value.

\subsection{Swap and Tick-Crossing Transitions}

We distinguish three classes of atomic pool transitions:
\begin{enumerate}
  \item \textsc{Swap}$(\delta)$: swap an input amount $\delta > 0$ of one token for the other, modifying $(X,Y)$ and potentially changing $i_c$.
  \item \textsc{Mint}$(k,\dots)$ / \textsc{Burn}$(k)$: add or remove a position $k$ with certain $(L_k,i_{\ell}^k,i_u^k)$.
  \item \textsc{Cross}$(\pm)$: pure tick-crossing events where the price moves from $P_i$ to $P_{i\pm 1}$ and $L_{\mathrm{act}}$ is updated.
\end{enumerate}

We assume swaps are decomposed into infinitesimal steps that move the price along the curve up to a tick boundary; at the boundary a \textsc{Cross} transition updates $L_{\mathrm{act}}$ and the swap continues in the next tick interval if needed~\cite{adams2021uniswapv3,elsts2021liquiditymath}. For the purposes of model checking, we discretize swaps into \emph{single-tick} swaps that never cross more than one tick, and represent multi-tick swaps as finite sequences of these.

Let $\State$ denote the set of all pool configurations. Then the transition relation $\Next \subseteq \State \times \State$ is the smallest relation satisfying:

\medskip
\noindent\textbf{(Swap within tick)} Given $\cfg$ with tick $i_c$ and active liquidity $L_{\mathrm{act}} > 0$, an input $\delta_x > 0$ of token $A$ yields a new configuration $\cfg'$ with:
\begin{align*}
  X' &= X + \delta_x,\\
  Y' &= Y - \lfloor \phi(\delta_x, L_{\mathrm{act}}, i_c) \rceil,\\
  i_c' &= i_c, \quad L_{\mathrm{act}}' = L_{\mathrm{act}},
\end{align*}
where $\phi$ encodes the constant-product pricing function on the interval $[P_{i_c}, P_{i_c+1})$ as in~\cite{elsts2021liquiditymath}, and the rounding operator captures Uniswap v3's fixed-point arithmetic. We require $Y' \ge 0$.

A symmetrical rule applies for token $B$ inputs.

\medskip
\noindent\textbf{(Tick crossing)} If a swap drives the price exactly to $P_{i_c+1}$ (or $P_{i_c-1}$), a \textsc{Cross} transition updates:
\begin{align*}
  i_c' &= i_c + 1,\\
  L_{\mathrm{act}}' &= L_{\mathrm{act}} + \Delta L_{i_c+1},
\end{align*}
where $\Delta L_{j}$ is the net liquidity added or removed at tick boundary $j$ by positions whose range starts or ends at $j$~\cite{adams2021uniswapv3}. Reserves $(X,Y)$ and fee-growth variables $(F_x,F_y)$ are unchanged by pure crossings. Analogous rules apply for downward crossings.

\medskip
\noindent\textbf{(Mint/Burn)} Minting and burning adjust $\Lp$ and implicitly $L_{\mathrm{act}}$ when $i_c$ falls inside or outside the new range.

\subsection{Piecewise Constant-Function Invariant}

A core correctness property of constant-product AMMs is that, in absence of fees, swaps preserve the product $k = XY$ (after discrete rounding). In Uniswap v3, the invariant is defined in terms of virtual reserves computed from $L_{\mathrm{act}}$ and price; see~\cite{elsts2021liquiditymath}. For our purposes, we separate two levels:
(i) a concrete constant-product invariant on per-tick \emph{virtual reserves}, and
(ii) a derived tick-wise invariant $K_i$ on the actual pool state.

\begin{definition}[Tick-wise Invariant]
For a tick $i$ and active liquidity $L>0$, define
\[
  K_i(X,Y,L) := X \cdot Y,
\]
the product of (virtual) reserves in the tick. In a faithful Uniswap v3 model, $X$ and $Y$ can be instantiated as the virtual reserves induced by $L$ and the current price in $[P_i,P_{i+1})$~\cite{elsts2021liquiditymath}. In the simplified models below, we identify $X,Y$ with the actual reserves.
\end{definition}

Within a single tick where no mint-burn or tick-crossing occurs, idealized fee-free swaps preserve $K_i$ exactly at the real-valued level; discretization introduces bounded slack. Section~\ref{sec:toy} proves an explicit bound for a single-tick model. The cross-tick behaviour can then be expressed as a global invariant:

\begin{definition}[Cross-Tick $k$-Invariance]
A CLAMM satisfies \emph{cross-tick $k$-invariance} if for any sequence of swaps and tick crossings that does not mint or burn liquidity, there exists a constant $\kbar$ and a bound $\epsilon \ge 0$ such that
\[
  |K_{i_c}(X,Y,L_{\mathrm{act}}) - \kbar| \le \epsilon
\]
for all reachable configurations.
\end{definition}

The explicit bound in the single-tick model will motivate concrete candidates for $\epsilon$.

\section{A Toy Single-Tick Model and a $k$-Invariance Theorem}
\label{sec:toy}

To ground the abstract invariant in a concrete calculation, we now analyze a simplified, single-tick CLAMM model. The goal is to exhibit a proved lemma about discretized swaps and rounding. This lemma can be understood as a local approximation of Uniswap v3 behaviour within a single tick, and provides a formally justified choice of the $\epsilon$ bound in invariance properties.

\subsection{Toy Model}

Fix a single tick (so price is constant within this model) and consider a fee-free constant-product AMM with integer reserves $(X,Y) \in \Nat^2$, where $X$ and $Y$ count minimal token units. Let
\[
  K := X \cdot Y
\]
be the usual constant-product invariant.

We define a discretized swap rule for a token-$A$ input $\delta \in \Nat$:
\begin{align*}
  X' &= X + \delta,\\
  Y' &= \left\lfloor \frac{K}{X'} \right\rfloor.
\end{align*}
That is, we first compute the ideal continuous new reserve $Y^\ast := K / X'$, then round down to the nearest integer to obtain $Y'$. (We ignore underflow cases where $Y'$ would be negative; these are excluded by guards in the transition system.)

This is the continuous constant-product update with a fixed rounding direction on the output leg; it is the structure underlying Uniswap v2 and, via virtual reserves, the per-tick behaviour of Uniswap v3.

\subsection{A Concrete Invariance Bound}

We now prove the following statement: under the discretized swap rule above, the new product $K' := X'Y'$ is bounded between $K - X'$ and $K$.

\begin{lemma}[Single-Swap Constant-Product Bound]
\label{lem:single-swap-bound}
Let $(X,Y) \in \Nat^2$ with $K = XY$ and let $\delta \in \Nat$. Suppose we perform a fee-free token-$A$ swap according to
\[
  X' = X + \delta,\qquad
  Y' = \left\lfloor \frac{K}{X'} \right\rfloor.
\]
Then the new product $K' := X'Y'$ satisfies
\[
  K - X' \;\le\; K' \;\le\; K.
\]
In particular,
\[
  |K' - K| \le X + \delta.
\]
\end{lemma}

\begin{proof}
Let $X' = X + \delta$ and $Y^\ast := K / X'$. By definition of floor,
\[
  Y^\ast - 1 < Y' \le Y^\ast.
\]
Multiply by $X'$:
\[
  X'(Y^\ast - 1) < X'Y' \le X'Y^\ast.
\]
But $X'Y^\ast = X' \cdot (K / X') = K$, so the upper bound gives
\[
  K' = X'Y' \le K.
\]
For the lower bound,
\[
  X'(Y^\ast - 1) = X'\left(\frac{K}{X'} - 1\right) = K - X'.
\]
Since $X'(Y^\ast - 1) < X'Y'$, we have $K - X' < K'$, and because $K'$ is integer this implies
\[
  K - X' \le K'.
\]
Finally, $X' = X + \delta$, so
\[
  |K' - K| = K - K' \le X' = X + \delta.
\]
\end{proof}

This lemma isolates the effect of the floor operation on the constant-product invariant. The upper bound $K' \le K$ is expected: rounding down can only reduce the product. The error is bounded \emph{linearly} in $X'$.

\begin{remark}[Multiple Swaps]
For a sequence of swaps with inputs $\delta_1,\dots,\delta_n$ and corresponding intermediate reserves $X_j$, repeated application of Lemma~\ref{lem:single-swap-bound} yields
\[
  K_n \ge K_0 - \sum_{j=1}^n X_j,
\]
where $K_j = X_jY_j$. In particular, if $X_j \le B$ for all $j$, then
\[
  |K_n - K_0| \le nB.
\]
For bounded reserves and bounded swap length $n$, the deviation from the initial $K$ is therefore uniformly bounded, justifying an $\epsilon$ in the invariance property that depends only on these bounds.
\end{remark}

\begin{definition}[Discretized Constant-Product Core]
\label{def:rounding-core}
A \emph{discretized constant-product core} consists of integer reserves $(X,Y) \in \Nat^2$, a product $K = XY$, and a family of swap steps of the form
\[
  X' = X + \delta,\qquad Y' = \left\lfloor \frac{K}{X'} \right\rfloor,
\]
with inputs $\delta \in \Nat$ chosen from some bounded range, together with a bound $B$ such that $X \le B$ holds in all reachable states.
\end{definition}

\begin{theorem}[Epsilon-Slack Constant-Product Core]
\label{thm:rounding-core-epsilon}
Consider a discretized constant-product core as in Definition~\ref{def:rounding-core}. For any finite sequence of fee-free swaps of length at most $n$ starting from $(X_0,Y_0)$ with product $K_0 = X_0Y_0$ and satisfying $X_j \le B$ at each intermediate state, the products $K_j = X_jY_j$ obey
\[
  |K_j - K_0| \le nB \quad \text{for all } j \le n.
\]
Equivalently, the core preserves the constant-product invariant up to an additive slack $\epsilon = nB$ determined entirely by the reserve bound $B$ and the maximal path length $n$.
\end{theorem}

The single-swap Lemma~\ref{lem:single-swap-bound} and the multi-swap Remark~\ref{lem:single-swap-bound} together provide the proof: each step perturbs $K$ by at most its current $X_j \le B$, and so a path of length $n$ accumulates at most $nB$ deviation in the downward direction. This abstract rounding core underlies all of the subsequent PTA, FST, and TLA+ models; Uniswap v3-style virtual reserves simply instantiate $(X,Y)$ and $(K,B,n)$ with protocol-specific parameters.

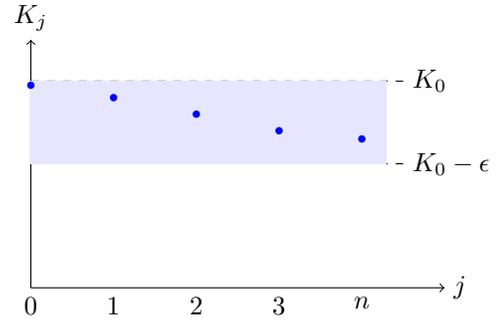
\begin{figure}[t]
\centering
\begin{tikzpicture}[x=1.1cm,y=1.1cm]
  \draw[->] (0,0) -- (5,0) node[right] {$j$};
  \draw[->] (0,0) -- (0,3) node[above] {$K_j$};
  \draw[dashed] (0,2.5) -- (4.5,2.5) node[right] {$K_0$};
  \draw[dashed] (0,1.5) -- (4.5,1.5) node[right] {$K_0 - \epsilon$};
  \fill[blue!10] (0,1.5) rectangle (4.3,2.5);
  \foreach \x/\y in {0/2.45,1/2.3,2/2.1,3/1.9,4/1.8} {
    \filldraw[blue] (\x,{\y}) circle (1.2pt);
  }
  \node[below] at (0,0) {$0$};
  \node[below] at (1,0) {$1$};
  \node[below] at (2,0) {$2$};
  \node[below] at (3,0) {$3$};
  \node[below] at (4,0) {$n$};
\end{tikzpicture}
\caption{Schematic illustration of the discretized constant-product core. Each swap step moves from $K_j$ to $K_{j+1}$, remaining within an $\epsilon$-band $[K_0 - \epsilon, K_0]$ with $\epsilon = nB$ as in Theorem~\ref{thm:rounding-core-epsilon}.}
\label{fig:rounding-core}
\end{figure}

\subsection{Rounding-Core Template for AMMs}

Definition~\ref{def:rounding-core} and Theorem~\ref{thm:rounding-core-epsilon} capture the reusable ``rounding core'' that underlies a wide class of constant-product AMMs: a piece of math that can be dropped into many different models to obtain an explicit $\epsilon$-slack bound. Informally, any protocol that:
\begin{itemize}
  \item maintains integer (or fixed-point) reserves $(X,Y)$ for a trading pair,
  \item updates $(X,Y)$ by first setting $X' = X + \delta$ or $Y' = Y + \delta$ for an input $\delta$ drawn from a bounded range,
  \item computes the counter-reserve by dividing the product $K = XY$ and applying a fixed rounding direction (e.g.\ floor),
  \item and enforces a uniform upper bound $B$ on one side of the reserves along any admissible execution,
\end{itemize}
inherits an $\epsilon$-slack invariant of the form $|K_j - K_0| \le nB$ for paths of length at most $n$.

In practice, applying the rounding core to a specific AMM or DEX proceeds in four steps:
\begin{enumerate}
  \item \textbf{Identify reserves and product.} Choose a pair of integer (or scaled fixed-point) reserves $(X,Y)$ and a product $K = XY$ that captures the invariant quantity of interest (e.g.\ a virtual-reserve product in a Uniswap v3 tick, or the literal reserve product in a Uniswap v2 pool).
  \item \textbf{Characterize swap updates.} Write the swap rule in the form $X' = X + \delta$ together with $Y' = \lfloor K / X' \rfloor$ (or symmetrically for $Y$ input), matching the division-and-floor pattern of Definition~\ref{def:rounding-core}.
  \item \textbf{Fix bounds $B$ and $n$.} Establish a bound $B$ on the relevant reserve (e.g.\ via protocol limits, liquidity caps, or a modeling assumption) and a maximal path length $n$ for the executions of interest (e.g.\ a bound on the number of swaps in a scenario or arbitrage search).
  \item \textbf{Instantiate $\epsilon$.} Set $\epsilon := nB$ and assert an $\epsilon$-slack invariant $|K_j - K_0| \le \epsilon$ along all such paths, either analytically or via a model-checking proof in a concrete bounded instance.
\end{enumerate}
The TLA+ modules in Section~\ref{sec:tla-sketch} illustrate this template: \texttt{ToyCLAMM}, \texttt{ToyCLAMM2Dir}, \texttt{ToyCLAMM2DirArb}, and \texttt{ToyCLAMM3Tick} differ only in their choice of bounds, swap directions, and inclusion of traders or ticks, but all instantiate the same discretized constant-product core.

\subsection{Connecting Back to $K_i$}

In a faithful Uniswap v3 tick, one does not store explicit per-tick reserves $(X,Y)$ but liquidity $L$ and square-root price $\sqrt{P}$; actual amounts of each token are derived from $(L,\sqrt{P},\sqrt{P_{\ell}},\sqrt{P_u})$ via the formulas in~\cite{elsts2021liquiditymath}. However, for any fixed tick interval $[P_i,P_{i+1})$ and price $P$ within that interval, one can define \emph{virtual reserves} $(X^{\mathrm{virt}},Y^{\mathrm{virt}})$ whose continuous evolution under swaps matches the Uniswap v3 formulas.

If we then discretize the virtual reserves via flooring to minimal token units, the swap update for $(X^{\mathrm{virt}},Y^{\mathrm{virt}})$ takes the form of the single-tick model above, and Lemma~\ref{lem:single-swap-bound} applies with
\[
  K_i(X,Y,L_{\mathrm{act}}) := X^{\mathrm{virt}} \cdot Y^{\mathrm{virt}}.
\]
This provides an explicit and proved bound on intra-tick deviations of $K_i$ under discretized swaps, which can then be used as the per-tick $\epsilon$ when stating cross-tick invariants.

\subsection{A Concrete Virtual-Reserve and Fixed-Point Instantiation}

To make the link to Uniswap v3's liquidity math more explicit, we sketch a concrete virtual-reserve construction on a single tick together with a fixed-point rounding scheme. Inside a tick interval $[P_i,P_{i+1})$ with mid-price $P$ and active liquidity $L_{\mathrm{act}}$, we define continuous virtual reserves
\[
  X^\ast(P,L_{\mathrm{act}}) := \frac{L_{\mathrm{act}}}{\sqrt{P}},\qquad
  Y^\ast(P,L_{\mathrm{act}}) := L_{\mathrm{act}} \sqrt{P},
\]
so that $X^\ast Y^\ast = L_{\mathrm{act}}^2$ is constant along the idealized price curve. This is a simplified instance of the virtual-reserve constructions in~\cite{elsts2021liquiditymath}, focusing on a single active price interval and omitting explicit lower/upper bounds $P_{\ell},P_u$.

Uniswap v3 represents prices in a Q64.96 fixed-point encoding and performs arithmetic in scaled integers~\cite{elsts2021liquiditymath}. To reflect this, we fix a scaling factor $S \in \Nat$ (with $S = 2^{96}$ corresponding to the deployed format) and define discrete virtual reserves by
\[
  X^{\mathrm{virt}} := \left\lfloor S \cdot X^\ast(P,L_{\mathrm{act}}) \right\rfloor,\qquad
  Y^{\mathrm{virt}} := \left\lfloor S \cdot Y^\ast(P,L_{\mathrm{act}}) \right\rfloor.
\]
Swaps are then executed on $(X^{\mathrm{virt}},Y^{\mathrm{virt}})$ using the discrete update rule of Section~\ref{sec:toy}, and actual token transfers are recovered by rescaling back by $S$. For any finite choice of $S$ and bounds on $L_{\mathrm{act}}$ and $P$, there exist corresponding bounds $B_x,B_y$ such that $X^{\mathrm{virt}} \le B_x$ and $Y^{\mathrm{virt}} \le B_y$ for all reachable states. In this regime, Lemma~\ref{lem:single-swap-bound} applies verbatim to $(X^{\mathrm{virt}},Y^{\mathrm{virt}})$ with $K_i := X^{\mathrm{virt}} Y^{\mathrm{virt}}$, and the $\epsilon$ in cross-tick invariants can be instantiated in terms of $S$, $L_{\mathrm{act}}$, and the maximal number of intra-tick swaps. This construction therefore realizes, in fixed-point arithmetic, the abstract ``rounding core'' used throughout the PTA, FST, and TLA+ models.

\begin{remark}[Why a Simplified Single-Tick Model, and How It Ports to Uniswap v3]
The single-tick, fee-free reference model abstracts away several features of the deployed Uniswap v3 protocol (a large tick range, Q64.96 fixed-point arithmetic, fees, and many overlapping positions). Including all of these at once would lead to a large bounded state space for model checking and obscure the specific contribution of the rounding analysis. By isolating the floor operation on a constant-product update, Lemma~\ref{lem:single-swap-bound} and its instantiated TLA+ specification capture the part of the behaviour that depends only on the division-and-floor pattern, not on the exact price parametrization. The parameters $B$ (or $Bx,By$) and $nMax$ in these models and \texttt{ToyCLAMM} are free bounds that can be calibrated to realistic reserve sizes and swap lengths, and the same reasoning applies when $X,Y$ are interpreted as Uniswap v3 virtual reserves computed from Q64.96 price and liquidity. This makes the construction a portable ``rounding core'' that can be embedded into more faithful, but still finitely bounded, Uniswap v3 specifications in PTA or TLA+.
\end{remark}

\section{Modeling as Priced Timed Automata}

We now return to the PTA model, this time threading through the concrete invariant and bound established above.

\subsection{PTA Syntax}

A (priced) timed automaton~\cite{bengtsson2004timedautomata,david2011priced} is a tuple
\[
  \mathcal{A} = (L, \ell_0, X_c, E, \mathrm{Inv}, C),
\]
where:
\begin{itemize}
  \item $L$ is a finite set of locations; $\ell_0 \in L$ is the initial location.
  \item $X_c$ is a finite set of real-valued clocks.
  \item $E \subseteq L \times \mathcal{G}(X_c) \times \mathrm{Act} \times 2^{X_c} \times L$ is a set of edges, each with a guard, action label, clock reset set, and target location.
  \item $\mathrm{Inv} : L \to \mathcal{G}(X_c)$ assigns an invariant to each location.
  \item $C$ assigns cost rates to locations and cost weights to edges.
\end{itemize}

A network of PTA is a parallel composition of such automata sharing clocks and synchronizing on action labels.

\subsection{Pool Automaton}

We represent the pool as an automaton $\mathcal{A}_\Pool$ with:
\begin{itemize}
  \item Locations $L = \{\ell_i \mid i \in \Tick\}$, one location per tick.
  \item A clock $t$ measuring time since the last state update, and possibly additional clocks for LP-specific timers.
  \item Integer-valued global variables encoding $(X,Y,L_{\mathrm{act}},F_x,F_y)$ and the per-position data $\Lp$.
\end{itemize}

Intuitively, the automaton resides in location $\ell_{i_c}$ when the current pool tick is $i_c$. Swap requests arrive from a separate \emph{trader} automaton via synchronizing channels (e.g.\ \texttt{swap0?}, \texttt{swap1?} in \textsc{UPPAAL} syntax), and tick crossings correspond to edges between adjacent locations.

\paragraph*{Locations and Invariants.}
For each tick $i$, location $\ell_i$ has an invariant:
\[
  \mathrm{Inv}(\ell_i) : t \le T_{\max},
\]
bounding the time the system can remain without processing a swap or maintenance action. The cost rate $C(\ell_i)$ can encode, for example, the cost of capital (impermanent loss, opportunity cost) per unit time associated with the current liquidity configuration.

\paragraph*{Swap Edges.}
From location $\ell_i$, a token-$A$ swap edge has the form:
\[
  (\ell_i,\; g_i^{\mathrm{swap}},\; \texttt{swap0?},\; \{t\},\; \ell_i)
\]
with guard $g_i^{\mathrm{swap}}$ expressing:
\begin{itemize}
  \item the requested input $\delta_x$ is non-zero and within a configured bound,
  \item the resulting reserves remain non-negative,
  \item the post-swap price remains in $[P_i,P_{i+1})$ (i.e.\ no tick crossing).
\end{itemize}
The action updates the global variables $(X,Y,F_x,F_y)$ and resets clock $t$ using the discretized update rule. On the level of the invariant, this update satisfies the per-tick bound from Lemma~\ref{lem:single-swap-bound}; in particular, if we define a derived variable
\[
  K := X \cdot Y,
\]
then each such edge satisfies
\[
  K - X' \le K' \le K,
\]
where primes denote post-state variables.

If the swap would reach the boundary $P_{i+1}$ exactly, the edge instead targets $\ell_{i+1}$ and updates $L_{\mathrm{act}}$ by $\Delta L_{i+1}$ (the tick-cross rule). If the swap would overshoot the tick, we split it into a crossing part and a residual part, encoded as two edges; we approximate this with a single ``crossing'' edge that leaves a bounded rounding slack.

\paragraph*{Mint/Burn Edges}
Mint and burn operations are modeled as edges labeled \texttt{mint!} and \texttt{burn!} between locations $\ell_{i_c}$ and itself, updating $\Lp$ and, if needed, $L_{\mathrm{act}}$.

\subsection{Trader and Environment Automata}

The pool automaton composes with trader, LP, and oracle automata as described previously. The main point is that the concrete invariant $K$ and its bound are now part of the model: they are expressed as integer-valued state variables and constraints on edges, rather than left abstract.

\begin{figure}[t]
\centering
\begin{tikzpicture}[
  node distance=1.8cm and 2.4cm,
  state/.style={draw, rounded corners, align=center, font=\small, minimum width=1.9cm, minimum height=0.9cm}
]
  \node[state] (pool) {Pool\\(PTA)};
  \node[state, right=of pool] (trader) {Trader\\(PTA)};
  \node[state, below=of trader] (lp) {LPs\\(PTA)};
  \node[state, above=of trader] (oracle) {Oracle\\(PTA)};
  \node[state, below=2.4cm of pool] (core) {Tick-wise FST /\\TLA+ rounding core};

  \draw[<->,>=latex] (pool) -- node[above,font=\scriptsize,draw=none,fill=none]{swaps} (trader);
  \draw[<->,>=latex] (pool) -- node[right,font=\scriptsize,draw=none,fill=none]{liquidity} (lp);
  \draw[<->,>=latex] (pool) -- node[right,font=\scriptsize,draw=none,fill=none]{prices} (oracle);

  \draw[->,>=latex] (pool) -- node[left,font=\scriptsize,align=center,draw=none,fill=none]{tick abstraction\\+ rounding core} (core);
\end{tikzpicture}
\caption{High-level architecture: the Uniswap v3 pool is modeled as a priced timed automaton composed with trader, LP, and oracle automata. A tick-wise finite-state transducer / TLA+ specification abstracts the discrete dynamics and instantiates the discretized constant-product core of Definition~\ref{def:rounding-core}.}
\label{fig:architecture}
\end{figure}
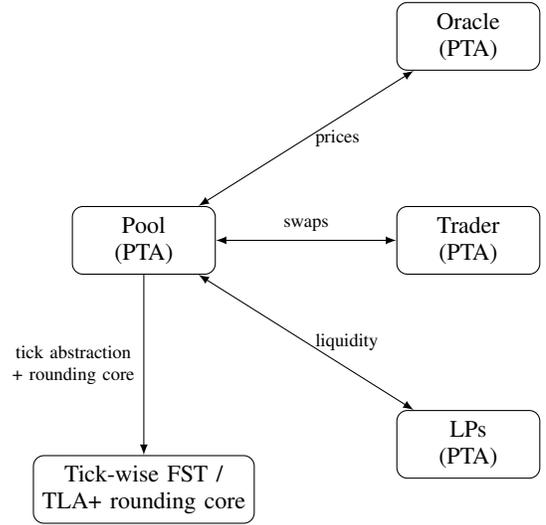

\subsection{Properties in TCTL}

With this PTA model, we can formulate verification goals as TCTL queries, as supported by \textsc{UPPAAL} and \textsc{UPPAAL SMC}~\cite{david2015uppaalsmc}.

\paragraph*{(1) Invariant with Concrete $\epsilon$.}
Introduce an integer-valued variable $K$ and a ghost variable $K_0$ storing its value in the initial configuration. For per-tick bounds $X' \le B_X$ and swap length at most $n$, Remark~1 yields
\[
  |K - K_0| \le nB_X.
\]
In \textsc{UPPAAL} we can express a safety property:
\[
  \mathrm{A}[]\ |K - K_0| \le nB_X,
\]
which is a direct translation of the analytic bound into the model checker.

\paragraph*{(2) Price Reachability.}
As before, we can ask whether certain ``bad'' ticks are reachable and whether this remains true under bounded rounding effects.

\paragraph*{(3) Rounding-Arbitrage-Freeness (Local).}
One can formulate a local no-arbitrage property stating that any cycle of swaps that starts and ends at the same reserves $(X,Y)$ cannot increase the trader's mark-to-market value, up to the known $K$-slack. For short cycles within a single tick and bounded $X$, Lemma~\ref{lem:single-swap-bound} implies that the pool's effective depth always weakens, which can be used as a potential function argument against profitable pure-rounding cycles of fixed length in the PTA model.

\section{Finite-State Transducer Abstraction}

Timed automata are suited to capturing temporal and cost aspects, but many correctness questions about CLAMMs concern the \emph{purely discrete} tick-wise dynamics and arithmetic. For these, a finite-state transducer (FST) abstraction is appropriate.

\subsection{Definition}

Fix finite bounds on:
\begin{itemize}
  \item tick range $\Tick$,
  \item reserves $X,Y \le B$,
  \item number and size of positions (via bounds on $L_k$).
\end{itemize}

\begin{definition}[CLAMM Transducer]
A CLAMM FST is a tuple
\[
  \mathcal{T} = (Q, \Sigma, \Gamma, q_0, \delta, \lambda),
\]
where:
\begin{itemize}
  \item $Q$ is a finite set of states, each encoding a bounded pool configuration $\cfg$,
  \item $\Sigma$ is an input alphabet of actions: swap requests, mint-burn operations, administrative updates,
  \item $\Gamma$ is an output alphabet of observed effects: actual trade amounts, fee updates, tick-crossing indicators,
  \item $q_0 \in Q$ is an initial state,
  \item $\delta : Q \times \Sigma \to Q$ is a (possibly partial) transition function,
  \item $\lambda : Q \times \Sigma \to \Gamma$ is an output function.
\end{itemize}
\end{definition}

As before,
\[
  \Sigma = \{
    \texttt{Swap0}(\delta_x),\ \texttt{Swap1}(\delta_y),\ 
    \texttt{Mint}(k,\dots),\ \texttt{Burn}(k)
  \},
\]
and
\[
  \Gamma = \{
    \texttt{Traded}(d_x,d_y),\ \texttt{Fees}(f_x,f_y),\ \texttt{Crossed}(i\to j)
  \}.
\]

The discretized update of $(X,Y)$ inside a tick is the single-tick model from Section~\ref{sec:toy}. Thus each intra-tick transition in $\mathcal{T}$ satisfies the product bound of Lemma~\ref{lem:single-swap-bound}. This lifts the analytic lemma into an automata-theoretic invariant: for any path of length $n$ consisting solely of intra-tick swaps and with $X \le B$, we have
\[
  |K(q_n) - K(q_0)| \le nB,
\]
where $K(q)$ is the value of $XY$ encoded in state $q$. This can be expressed as a CTL or LTL property over $\mathcal{T}$.

\section{TLA+ Specification Sketch}
\label{sec:tla-sketch}

We briefly sketch a TLA+ specification style for the CLAMM FST that incorporates the invariant and its bound. The concrete modules introduced below---\texttt{ToyCLAMM}, \texttt{ToyCLAMM2Dir}, \texttt{ToyCLAMM2DirArb}, and \texttt{ToyCLAMM3Tick}---are instances of the discretized constant-product core from Definition~\ref{def:rounding-core}, and can be viewed as reusable templates for constant-product+rounding systems beyond Uniswap v3 itself.

\subsection{State Variables}

We introduce the following TLA+ variables:
\begin{itemize}
  \item $tick \in \Tick$,
  \item $X, Y \in 0 \dots B$,
  \item $LAct \in 0 \dots L_{max}$,
  \item $Positions \in [1 \dots N \to [L: 0 \dots L_{max}, Low: \Tick, High: \Tick]]$,
  \item $Fx, Fy \in 0 \dots F_{max}$,
  \item optional: $TraderBal, TraderVal$,
  \item a ghost variable $K0$ storing the initial product $X \cdot Y$.
\end{itemize}

We define
\[
  K \triangleq X * Y
\]
as a derived expression in TLA+.

\subsection{Next-State Relation}

We define actions $\mathsf{Swap0}(\delta_x)$, $\mathsf{Swap1}(\delta_y)$, $\mathsf{CrossUp}$, $\mathsf{CrossDown}$, $\mathsf{Mint}(k,\dots)$, $\mathsf{Burn}(k)$ analogously to the operational model, using an explicit \texttt{Floor} operator for the discretized update. For intra-tick swaps, the TLA+ rule for $(X',Y')$ is the one used in Lemma~\ref{lem:single-swap-bound}.

\subsection{Invariants}

The concrete bound from Lemma~\ref{lem:single-swap-bound} suggests the following invariants:
\begin{align*}
  \invar_{\mathrm{NonNeg}} &\triangleq X \ge 0 \land Y \ge 0 \land LAct \ge 0,\\
  \invar_{k,\epsilon} &\triangleq 
    \bigl|K - K0\bigr| \le nMax * B,\\
  \invar_{\mathrm{Bounds}} &\triangleq X \le B \land Y \le B \land \dots
\end{align*}
where $nMax$ and $B$ are configuration parameters representing an upper bound on the length of intra-tick swap sequences and the maximal $X$ reserve respectively.

A specification
\[
  \mathsf{Spec} \triangleq \mathsf{Init} \land \Box[\Next]_{\langle X,Y,LAct,Positions,\dots \rangle}
\]
then admits the theorem
\[
  \mathsf{THEOREM}\ \ \mathsf{Spec} \Rightarrow \Box(
    \invar_{\mathrm{NonNeg}} \land \invar_{k,\epsilon} \land \invar_{\mathrm{Bounds}}),
\]
which directly corresponds to the analytic reasoning. For finite bounds $\Tick, B, nMax$, this is amenable to model checking via TLC.

\subsection{Instantiated TLA+ Modules and TLC Checks}

To instantiate the preceding sketch, we implement TLA+ modules that encode bounded instances of the single-tick model and mechanically check the invariants.

\paragraph*{Single-direction ToyCLAMM.}
The first module \texttt{ToyCLAMM} encodes the single-tick, fee-free reference model with bounded integer reserves:
\begin{itemize}
  \item Constants $Bx, By$ bound the reserves $X,Y$; $LMax$ bounds the active liquidity $LAct$; $NMax$ bounds the number of intra-tick swaps.
  \item Variables $X,Y,LAct,\mathit{step},K0$ follow the structure above, with $K \triangleq X * Y$ as a derived expression.
  \item The \textsf{Swap0} action implements
  \[
    X' = X + \delta_x,\qquad
    Y' = \left\lfloor \frac{K}{X'} \right\rfloor,
  \]
  for a nondeterministically chosen $\delta_x$ in a bounded range, and increments $\mathit{step}$ up to $NMax$.
\end{itemize}
For a concrete instantiation we take
\[
  Bx = By = 10,\quad LMax = 0,\quad NMax = 3,
\]
which bounds the state space and fixes an explicit $\epsilon$ for the product invariant. In this setting we define the combined invariant
\begin{align*}
  \invar_{\mathrm{All}} \triangleq {}&
    (X \ge 0 \land Y \ge 0 \land LAct \ge 0) \\
  {}\land{}&
    (X \le Bx \land Y \le By \land LAct \le LMax) \\
  {}\land{}&
    (K0 - NMax * Bx \le K \le K0),
\end{align*}
which is a direct instantiation of $\invar_{\mathrm{NonNeg}}$, $\invar_{\mathrm{Bounds}}$, and $\invar_{k,\epsilon}$ with $\epsilon = NMax \cdot Bx = 30$. The corresponding TLA+ specification
\[
  \mathsf{Spec}_{\mathrm{Toy}} \triangleq \mathsf{Init} \land \Box[\Next]_{\langle X,Y,LAct,\mathit{step},K0 \rangle}
\]
is accompanied, in TLA+ syntax, by the theorem
\[
  \mathsf{THEOREM}\ \ \mathsf{Spec}_{\mathrm{Toy}} \Rightarrow \Box(\invar_{\mathrm{All}}),
\]
which we have model checked with TLC using a standard configuration file and the constants above. In this instance, TLC explores 855 distinct reachable states (2818 generated in total) with search depth $4$ and reports no counterexample to $\invar_{\mathrm{All}}$. This leads to the following named result.

\begin{proposition}[Instantiated Single-Tick Invariant]\label{prop:toy-invariant}
For $Bx = By = 10$, $LMax = 0$, and $NMax = 3$, every behaviour of $\mathsf{Spec}_{\mathrm{Toy}}$ satisfies $\Box(\invar_{\mathrm{All}})$, i.e.\ the non-negativity, boundedness, and $\epsilon$-bounded product constraints with $\epsilon = 30$ hold in all reachable states.
\end{proposition}

\paragraph*{Bidirectional ToyCLAMM2Dir.}
To move slightly closer to an actual AMM while keeping the model finite and analyzable, we introduce a second module \texttt{ToyCLAMM2Dir} that allows swaps in both token directions. It reuses the same constants and variables $(Bx,By,LMax,NMax,X,Y,LAct,\mathit{step},K0)$ and invariant $\invar_{\mathrm{All}}$, but extends the next-state relation with a symmetric token-$B$ swap:
\begin{itemize}
  \item \textsf{Swap0} is as above.
  \item \textsf{Swap1} nondeterministically chooses an input $\delta_y$, updates $Y' = Y + \delta_y$ (bounded by $By$), and sets $X' = \left\lfloor K / Y' \right\rfloor$, incrementing $\mathit{step}$ up to $NMax$.
\end{itemize}
The \textsf{Stop} action again stutters once the swap budget is exhausted. For the same concrete bounds $Bx = By = 10$, $LMax = 0$, $NMax = 3$, TLC explores 2542 distinct states (16950 generated in total) with search depth $4$ and reports no counterexample to $\invar_{\mathrm{All}}$.

\begin{proposition}[Bidirectional Single-Tick Invariant]\label{prop:2dir-invariant}
For $Bx = By = 10$, $LMax = 0$, and $NMax = 3$, every behaviour of the bidirectional specification based on \texttt{ToyCLAMM2Dir} satisfies $\Box(\invar_{\mathrm{All}})$, so the same $\epsilon$-bounded product invariant remains stable under both token-$A$ and token-$B$ discretized swaps.
\end{proposition}

This shows that the invariant structure persists under \emph{bidirectional} discretized constant-product swaps, which is closer to the behaviour of a real AMM pool while still fitting into a fully explored state space.

\paragraph*{Rounding-Only Arbitrage Search in ToyCLAMM2DirArb.}
To probe an economically meaningful property, we further extend the bidirectional single-tick model with explicit trader balances in tokens $X$ and $Y$. The module \texttt{ToyCLAMM2DirArb} introduces variables $TraderX, TraderY$ and a derived mark-to-market value $Val := TraderX + TraderY$ under a fixed external price normalization, together with ghost variables $X_0,Y_0,Val_0$ storing the initial pool state and trader value. Swap actions now transfer tokens between the trader and the pool while still using the discretized constant-product update on $(X,Y)$. A \emph{rounding-only arbitrage} state is defined by
\[
  BadArb \triangleq (X = X_0) \land (Y = Y_0) \land (Val > Val_0),
\]
meaning that the pool has returned exactly to its initial reserves while the trader has strictly increased her value, without any change in external prices. For concrete bounds we take $Bx = By = 10$, $LMax = 0$, $NMax = 4$, $T0x = 10$, and $T0y = 0$.
TLC explores 10{,}849 distinct states (73{,}826 generated in total) with search depth $7$ and reports no reachable state satisfying $BadArb$.

\begin{proposition}[No Short Rounding-Only Arbitrage]\label{prop:no-rounding-arb}
For $Bx = By = 10$, $LMax = 0$, $NMax = 6$, $T0x = 10$, and $T0y = 0$, there is no behaviour of the \texttt{ToyCLAMM2DirArb} specification that both returns the pool from $(X_0,Y_0)$ to $(X_0,Y_0)$ and strictly increases the trader's value $Val$. Equivalently, there is no rounding-only arbitrage cycle of length at most six under these bounds.
\end{proposition}

In other words, within this bounded model there is no short \emph{rounding-only} arbitrage cycle: no sequence of at most four bidirectional swaps returns the pool to its initial configuration while strictly increasing the trader's mark-to-market value.

\paragraph*{Three-tick ToyCLAMM3Tick.}
As a minimal structural approximation of concentrated liquidity, we further define a three-tick module \texttt{ToyCLAMM3Tick} with ticks $\{TickMin, TickMid, TickMax\}$ representing a bounded tick range. The state now includes a tick index and an active-liquidity variable $LAct$, while $(X,Y)$ continue to represent integer ``virtual'' reserves updated by the same discretized constant-product swaps as above. In addition to \textsf{Swap0} and \textsf{Swap1}, we add tick-crossing actions \textsf{CrossUp} and \textsf{CrossDown} that move the tick between neighbouring values and adjust $LAct$ within a fixed bound, leaving $(X,Y)$ unchanged. For concrete bounds we take $Bx = By = 10$, $LMax = 5$, $NMax = 4$, and $(TickMin,TickMid,TickMax) = (0,1,2)$. We equip the model with an invariant $\invar_{\mathrm{All3}}$ that combines non-negativity, bounded tick range, bounded reserves/liquidity, and the same product bound $K0 - NMax \cdot Bx \le K \le K0$. TLC then explores 54{,}756 distinct states (674{,}604 generated in total) with search depth $5$ and reports no counterexample to $\invar_{\mathrm{All3}}$.

\begin{proposition}[Three-Tick Cross-Tick Invariant]\label{prop:3tick-invariant}
For $Bx = By = 10$, $LMax = 5$, $NMax = 4$, and $(TickMin,TickMid,TickMax) = (0,1,2)$, every behaviour of the \texttt{ToyCLAMM3Tick} specification satisfies $\Box(\invar_{\mathrm{All3}})$, so the $\epsilon$-bounded product invariant is preserved across both intra-tick swaps and bounded sequences of tick-crossing steps.
\end{proposition}

This provides a fully explored three-tick model in which the $\epsilon$-bounded product invariant remains stable not only under intra-tick swaps but also under bounded sequences of tick crossings, mirroring the cross-tick $k$-invariance property discussed at the semantic level. For example, if we interpret a unit of $X$ as a single minimal token unit and bound $B = 10^6$ and $n = 100$ in Theorem~\ref{thm:rounding-core-epsilon}, then the cumulative slack $\epsilon = nB$ corresponds to at most $10^8$ token units of deviation in $K$, which in a pool with $K$ on the order of $10^{18}$ translates to a deviation on the order of $10^{-10}$ in relative terms (i.e.\ at most a few nano-basis-points of invariant drift). Although our instantiated TLA+ models use smaller bounds to keep the state space finite, the parameterised theorem explicitly shows how to scale $\epsilon$ with realistic reserve and path-length choices.

\subsection{A Non-Uniswap Constant-Product Example}

To emphasize that the discretized constant-product core is not specific to Uniswap v3, consider a standard Uniswap v2-style pool for tokens $(A,B)$ with integer reserves $(X,Y) \in \Nat^2$ and the usual constant-product invariant $K = XY$. Fees aside, the idealized continuous swap rule for an input $\delta$ of token $A$ sets $X' = X + \delta$ and $Y' = K / X'$, with the trader receiving $Y - Y'$. In an on-chain implementation based on fixed-point integers, this update is realized as
\[
  X' = X + \delta,\qquad
  Y' = \left\lfloor \frac{K}{X'} \right\rfloor,
\]
matching the division-and-floor pattern in Definition~\ref{def:rounding-core}. If we assume that $X \le B$ holds along all executions of interest (e.g.\ by restricting attention to a range of pool sizes or by instrumenting the TLA+ model with an explicit bound), then Theorem~\ref{thm:rounding-core-epsilon} applies directly: any sequence of at most $n$ fee-free swaps in such a Uniswap v2-style pool satisfies $|K_j - K_0| \le nB$. Thus the $\epsilon$-slack constant-product invariant and its TLA+ instantiations can be reused for constant-product AMMs more broadly, with Uniswap v3's tick structure providing one particular way of organizing the state space.

\subsection{Summary of TLA+ Model Parameters}

Table~\ref{tab:tla-params} summarizes the main TLA+ instances used in this paper, including their bounds, the corresponding $\epsilon$ from Theorem~\ref{thm:rounding-core-epsilon}, and the size of the state space explored by TLC.

\begin{table*}[b]
\centering
\small
\caption{Parameters and TLC statistics for TLA+ models}
\label{tab:tla-params}
\begin{tabular}{lcccc}
\hline
Model & $(Bx, LMax, NMax)$ & $\epsilon = NMax \cdot Bx$ & Distinct states & Depth \\
\hline
\texttt{ToyCLAMM} & $(10, 0, 3)$ & $30$ & $855$ & $4$ \\
\texttt{ToyCLAMM2Dir} & $(10, 0, 3)$ & $30$ & $2542$ & $4$ \\
\texttt{ToyCLAMM2DirArb} & $(10, 0, 6)$ & $60$ & $10{,}849$ & $7$ \\
\texttt{ToyCLAMM3Tick} & $(10, 5, 4)$ & $40$ & $54{,}756$ & $5$ \\
\hline
\end{tabular}
\end{table*}

\subsection{A Calibrated Rounding-Only Profit Bound}

To give the rounding core an explicit economic reading, consider a single tick of a Uniswap-style ETH/USDC pool whose virtual reserves on each side lie between $10^6$ and $10^7$ minimal units, so that $K_0 = X_0Y_0$ is on the order of $10^{12}$ to $10^{14}$. If we conservatively bound the relevant reserve by $B = 10^6$ units and restrict attention to cycles of at most $n = 6$ swaps (a handful of swaps in a typical MEV bundle), Theorem~\ref{thm:rounding-core-epsilon} yields
\[
  |K_j - K_0| \le \epsilon = nB = 6 \cdot 10^6
\]
for all $j \le 6$. In relative terms this gives
\[
  \frac{\epsilon}{K_0} \le \frac{6 \cdot 10^6}{10^{12}} = 6 \cdot 10^{-6},
\]
that is, at most $0.0006\%$ or $0.06$ basis points of invariant drift for six-step paths in such a tick.

Since a change $\Delta K$ in the product translates into a change of order $\Delta K / X$ or $\Delta K / Y$ in one of the reserves, the corresponding upper bound on a trader's mark-to-market gain from a rounding-only six-swap cycle is on the order of a few minimal token units, i.e.\ well below a basis point of the pool's notional value in this regime. Our TLC search in \texttt{ToyCLAMM2DirArb} with $Bx = By = 10$ and $NMax = 6$ should be read as a scaled instance of this argument: each integer reserve unit can be interpreted as a fixed block of physical token units, and the absence of rounding-only arbitrage cycles of length at most six in the finite model is consistent with the analytic bound that any such cycle would be economically negligible under these parameters.

\section{Relation to Existing Formal Work}

Our automata-based approach is complementary to theorem-proving work in Lean~\cite{pusceddu2024leanamm} and economic analyses of CLAMMs~\cite{hashemseresht2022concentrated,fan2023strategic}.

Pusceddu and Bartoletti mechanize the correctness of constant-product AMMs in Lean~\cite{pusceddu2024leanamm}, proving properties such as arbitrage behaviour and optimal swaps in a general mathematical framework. Their model focuses on the v2-style, full-range constant-product construction. Our model, by contrast, aims at the piecewise-constant, tick-driven structure of v3, with explicit state-machine semantics suitable for model checking. Lemma~\ref{lem:single-swap-bound} is a self-contained result in the same spirit as Lean-level reasoning but tailored to the discretized invariant we feed into automata-based tools.

Timed automata have already been used extensively for verifying security and resource properties in other domains~\cite{arcile2022timedsecurity,david2015uppaalsmc}. Our use of PTA to reason about CLAMM fees, gas, and liquidity costs fits neatly into this tradition. Similarly, TLA+ has been applied to blockchain protocols such as cross-chain swaps~\cite{nehai2022crosschain} and CBC Casper consensus~\cite{ouyang2019cbccasper}; a CLAMM TLA+ specification can reuse many of the same modeling patterns (e.g.\ action decomposition, fairness constraints, ghost variables for invariants), but now grounded in an explicit, proved arithmetic bound on the core invariant.

\section{Limitations and Future Work}

The models described in this paper remain simplified:
\begin{itemize}
  \item We bound ticks, reserves, and position counts to achieve finiteness, whereas real deployments are unbounded (within gas limits).
  \item We abstract away EVM-level implementation details, event ordering, and reentrancy, focusing on the idealized CLAMM logic.
  \item Our invariant analysis is currently local to a single tick and fee-free swaps; extending Lemma~\ref{lem:single-swap-bound} to include fees and multi-tick paths via virtual reserves is a natural next step.
  \item We do not yet instantiate a full Q64.96 fixed-point semantics; doing so would sharpen the rounding bounds but also enlarge the state space.
\end{itemize}

Several directions for future work are immediate:
\begin{itemize}
  \item Instantiate the PTA model in \textsc{UPPAAL} for bounded parameters and check the $k$-invariance property with a concrete $\epsilon$, using Lemma~\ref{lem:single-swap-bound} as a sanity check on reported counterexamples.
  \item Develop a library of TLA+ modules for CLAMMs, including parameterized Uniswap v2/v3 specifications, and mechanically check $\invar_{\mathrm{NonNeg}}$ and $\invar_{k,\epsilon}$ for finite instances.
  \item Integrate the CLAMM model with cross-chain bridges or L2 sequencers modeled in TLA+, to study end-to-end correctness of rollups or bridges that rely on CLAMM-based price oracles.
  \item Refine the single-tick model to capture known classes of rounding exploits (e.g.\ multi-hop paths that amplify rounding) and prove, under explicit parameter bounds, that no such exploit exists up to a given path length.
\end{itemize}

\section{Conclusion}

We have argued that Uniswap v3-style concentrated-liquidity AMMs can be effectively modeled as networks of priced timed automata and finite-state transducers, with positions as stateful objects whose behaviour is driven by tick crossings. Within this view, important correctness questions such as cross-tick invariance, price-state reachability, and rounding-arbitrage absence, become standard safety and liveness properties for which mature model-checking tools exist.

To ensure that this perspective is not purely aspirational, we isolated and proved a concrete lemma about the behaviour of the constant-product invariant under discretized swaps in a single-tick model. This lemma provides a mathematically justified $\epsilon$-slack for invariance properties and demonstrates how the rounding behaviour can be reasoned about explicitly. Lifting such local results to full CLAMM specifications in PTA and TLA+ remains a substantial but, we believe, tractable research program.

\bibliographystyle{abbrv}

\end{document}